\pdfoutput=1

\documentclass[12pt]{article}

\usepackage{natbib}
\newcommand{\cites}[1]{\citeauthor{#1}'s (\citeyear{#1})}
\usepackage[pdftex]{hyperref}
      
\usepackage[pdftex]{graphicx}
\usepackage{subfig}
\usepackage[section]{placeins}

\usepackage[flushleft]{threeparttable}

\usepackage{asymptote}

\usepackage{amsmath}
\usepackage{amsthm}
\newtheorem*{lemma*}{Lemma}
\newtheorem*{example}{Example}

\usepackage[mathscr]{euscript}
\usepackage{bm}
\usepackage{mathtools}

\usepackage[margin=1in]{geometry}
\usepackage[varg]{txfonts}

\usepackage{color}

\usepackage{listings}
\newcommand{\E}{\mathbf{E}}
\newcommand{\Var}{\mathrm{Var}}

\newcommand{\1}{\mathbf{1}}

\DeclareMathDelimiter{\Norm}{\mathopen}{largesymbols}{"42}{largesymbols}{"42}

\DeclareMathOperator*{\argmax}{argmax}

\newcommand{\indep}{\rotatebox[origin=c]{90}{$\models$}}

\begin{document}

\title{Estimating the Variance of Measurement Errors in Running Variables of Sharp Regression Discontinuity Designs}
\author{Kota Mori \thanks{kmori05@gmail.com}}

\begin{titlepage}

\maketitle

\begin{abstract}
\noindent
Estimation of a treatment effect by a regression discontinuity design faces a severe challenge when the running variable contains measurement errors since the errors smoothen the discontinuity on which the identification depends.
The existing studies show that the variance of the measurement errors plays a vital role in both bias correction and identification under such situations.
However, the methodologies to estimate the variance from data are relatively undeveloped.
This paper proposes two estimators for the variance of measurement errors of running variables of sharp regression continuity designs.
The proposed estimators can be constructed merely from data of the observed running variable and treatment assignment, and do not require any other external source of information.
\end{abstract}
\thispagestyle{empty}
\end{titlepage}

\section{Introduction}

Regression discontinuity design (RDD) is a frequently-used framework for estimating the causal effect of a binary treatment variable on an outcome measurement.
An RDD depends on a critical assumption that there exists a variable such that the treatment is assigned if and only if that variable exceeds a known threshold.
A variable with this property is called a running variable.
%
Given an RDD framework, one compares the treated and untreated samples around the threshold of the running variable.
Assuming that other covariates are continuously distributed at that point, those slightly above the threshold and those slightly below are arguably similar except that only the former receives the treatment. 
Therefore the difference in the outcome measurement between the two is attributable to the impact of treatment.

Identification using an RDD faces a challenge when the observed running variable contains measurement errors.
Theoretically, even a small magnitude of measurement error would nullify the estimation of the treatment effect leveraging an RDD.
This is because the measurement errors smooth out the discontinuity of the assignment at the threshold, which breaks the RDD assumption.
Note that an RDD with a mismeasured running variable does not form a fuzzy RDD;
a fuzzy RDD assumes that the assignment probability is discontinuous at a threshold, while measurement errors of the running variable smoothen the discontinuity.

\citet{Davezies2014} showed that the standard local polynomial regression yields a biased estimate for the treatment effect if the running variable is mismeasured.
They then proposed an alternative estimator that is less susceptible to the measurement errors and examined the magnitude of the bias.
\citet{Yanagi2014} also studied a similar estimator and proposed a method to alleviate the bias of the estimator.
Finally, \citet{Pei2017} proposed a series of identification strategies that overcome measurement errors in the running variable.  

These studies agree that the variance of the measurement errors plays an essential role in the bias correction as well as in the identification of the treatment effect.
The analysis by \citet{Davezies2014} shows that their estimator would be more biased when the running variable contains measurement errors of a larger magnitude.
\cites{Yanagi2014} bias correction approach requires that the variance be known from an external source.
One of the estimators proposed by \citet{Pei2017} also utilizes external knowledge of the variance (see Approach 3 in \S 4.1).
Despite its utility, only a handful of discussions have been devoted to how one can obtain or estimate the variance of the measurement errors.
\citet{Yanagi2014} suggests that the variance can be estimated using auxiliary data that provide the accurate distribution of the running variable (but are not tied with the treatment assignment).
If such data are available, the variance of the measurement errors can be estimated by subtracting the {\em true} variance of the running variable from the variance of the mismeasured running variable.
Such auxiliary data, however, might not be available in many applications.

This paper proposes two estimators for the variance of the measurement errors.  
Both estimators do not require any additional source of information;
the estimation only requires data of the observed running variable and treatment assignment, which are naturally available in virtually all RDD studies. 
The first estimator assumes that both the running variable and the measurement error follow the Gaussian distribution.
Under this assumption, the conditional likelihood function has an analytic formula, which can be optimized efficiently by standard numerical methods.
The second estimator relaxes the Gaussian assumption and allows both the running variable and measurement error to follow arbitrary distributions characterized by a finite number of parameters.  
Unlike the Gaussian case, the likelihood function under this assumption cannot be expressed by a simple formula, where direct optimization becomes numerically unstable.
Instead, the likelihood can be maximized by a variant of the expectation-maximization (EM) algorithm, which is computationally efficient and robust.

The result of the simulation experiments are also reported.
All estimators successfully recover the true variance when the model assumption matches the data generation process.  
The estimators exhibit different degrees of robustness against misspecification.
The methods have been implemented as a library for the R language \citep{RCoreTeam2017} and are freely available on the GitHub repository (\url{https://github.com/kota7/rddsigma}).

\section{Model}

Let $D \in \{0,1\}$ denote the binary variable that indicates the assignment of treatment and $X \in \mathbf{R}$ the running variable for $D$. 
Suppose that $X$ and $D$ form a sharp regression discontinuity design, {\it i.e.}, $D = \1\{X>c\}$, with a known constant $c$.

Assume that $X$ is only observed with an additive error:
\[ W = X + U, \;\;\;\;\;\;  X \indep U \]
where $W$ is the observed running variable, for which data are available.
We assume that $U$ is continuous, has a zero mean and a finite variance $\sigma^2$.
Our goal is to estimate $\sigma$ using a random sample of $\{w_i, d_i\}_{i=1}^{n}$, where $w_i$ and $d_i$ represent the observations corresponding to $W$ and $D$ respectively.

\subsection{Gaussian-Gaussian Case}

Consider a case where both $X$ and $U$ follow the Gaussian distribution.
The independence assumption of the two implies that they follow the multivariate Gaussian distribution.
Therefore, the sum of the two, $W$, is also Gaussian.

Let $\E(X) = \mu_x$ and $\Var(X) = \sigma_{x}^{2}$.  Then, 
$\E(W) = \mu_x$ and 
$\Var(W) = \sigma_{x}^{2} + \sigma^2 \equiv \sigma_{w}^{2}$.
By the property of the multivariate Gaussian distribution \citep[see \it{e.g.},][\S 2.3.1]{Bishop2006}, the conditional distribution of $U$ given $W$ is also Gaussian and its parameters can be explicitly written as follows: 
\begin{equation}
\E(U | W) = \mu_{u|w} = \frac{\sigma^2}{\sigma_{w}^{2}}(W - \mu_x) 
\label{eq:condmean}
\end{equation}
and 
\begin{equation}
 \Var(U | W) = \sigma_{u|w}^2 =
   \left(1-\frac{\sigma^2}{\sigma_{w}^{2}} \right)\sigma^2. 
\label{eq:condvar}
\end{equation}
We can construct the conditional likelihood function using \eqref{eq:condmean} and \eqref{eq:condvar}.
Consider $p(D | W; \theta)$, that is, the conditional distribution of $D$ given $W$, where $\theta = (\mu_x, \sigma_w, \sigma)$.
Since $D = \1\{X>c\} = \1\{U < W - c\}$, we have
\[ p(D| W; \theta) = 
  \begin{cases}
      \Phi((W - c - \mu_{u|w})/\sigma_{u|w})  & \text{if $D = 1$}  \\
      1 - \Phi((W - c - \mu_{u|w})/\sigma_{u|w})  & \text{if $D = 0$}
  \end{cases}
\]
where $\Phi$ is the cumulative distribution function of the standard Gaussian distribution.

Although the likelihood function depends on three parameters, $(\mu_x, \sigma_w, \sigma)$, the first two can be estimated separately by the sample mean and standard deviation of $W$.
We can substitute these estimates into the likelihood function, and estimate $\sigma$ by the maximum likelihood.
Notice that this estimation process is a two-step maximum likelihood, and hence the variance of estimators needs to be adjusted appropriately \citep{Murphy1985, Newey1994}.

\subsection{Non-Gaussian Case}

In this section, we relax the Gaussian assumption in the previous section.
Assume instead that $X$ and $U$ follow some parametric distributions characterized by a finite number of parameters.
Unlike the Gaussian case, we do not have an explicit expression for the conditional likelihood under this assumption in general.
Instead, we consider the estimation using the marginal likelihood function.

Let $p_x$ and $p_u$ denote the probability density functions of $X$ and $U$ and suppose that they depend on parameters $\theta_x$ and $\theta_u$ respectively.
We can write the full likelihood function for a pair $(W, D)$ as 
\[
\begin{split}
\log p(W, D; \theta) &= 
D \log \int_{c}^{\infty} p_x(x; \theta_x) p_u(W-x; \theta_u) dx\\ +
&(1-D) \log \int_{-\infty}^{c} p_x(x; \theta_x) p_u(W-x; \theta_u) dx
\end{split}
\]
Our objective is to maximize the sum of log-likelihood with respect to the parameters:
\[
\hat{\theta} \equiv \argmax_\theta \sum_{i=1}^{n} \log p(w_i, d_i; \theta)
\]

Due to the complex expressions inside integrals, the direct maximization of this objective function by numerical routines tends to be computationally demanding and unstable.
Instead, we employ a variant of the expectation-maximization (EM) algorithm, which turns out to be computationally more efficient and robust.
Define the Q-function as below.
\[
\begin{split}
Q(\theta, \theta' | W, D) &= 
D \int_{c}^{\infty} h(\theta' | Z, D) 
  \left( \log p_x(x; \theta_x) + \log p_u(W-x; \theta_u) \right) dx\\ +
&(1-D) \int_{\infty}^{c} h(\theta' | Z, D) 
  \left( \log p_x(x; \theta_x) + \log p_u(W-x; \theta_u) \right) dx 
\end{split}
\]

where the function $h$ is defined as
\begin{align}
h(\theta, x | W, D=1) &= 
  \frac{p_x(x; \theta_x) p_u(W-x; \theta_u)}
  {\int_{c}^{\infty} p_x(x; \theta_x) p_u(W-x; \theta_u) dx} \label{eq:h1}\\
h(\theta, x | W, D=0) &= 
  \frac{p_x(x; \theta_x) p_u(W-x; \theta_u)}
  {\int_{\infty}^{c} p_x(x; \theta_x) p_u(W-x; \theta_u) dx}.  \label{eq:h2}
\end{align}

We can show that, for any $(W, D)$ and $\theta, \theta'$, 
\begin{equation} \label{eq:ineq}
\log p(W, D; \theta) - \log p(W, D; \theta') \ge
Q(\theta, \theta' | W, D) - Q(\theta', \theta' | W, D).
\end{equation}
See the Appendix \ref{sec:proof} for the proof of this inequality.

The inequality \eqref{eq:ineq} motivates a variant of the EM algorithm where the parameters are updated so as to maximize the sum of the Q-functions:  
\begin{equation}
\theta^{(t+1)} \leftarrow \argmax_\theta \sum_{i=1}^{n} Q(\theta, \theta^{(t)} | w_i, d_i), \label{eq:maxstep}
\end{equation}
where $\theta^{(0)}$ is initialized outside the loop.
By the inequality \eqref{eq:ineq}, the objective function increases monotonically along with the iterations, and hence converges to a local maximum provided that it is bounded. 
Note that, since the algorithm only ensures the convergence to a local maximum, the outcome may vary by choice of the initial value, $\theta^{(0)}$.

Iterative maximization of the Q-function tends to be computationally more efficient and stable than maximizing the likelihood function directly.
In particular, for the distributions such that the maximum likelihood parameter estimator is analytically solvable, the update \eqref{eq:maxstep} also has a closed-form expression.
We illustrate a case where $X$ follows the Gaussian distribution and $U$ the Laplace distribution below.

\begin{example}

Suppose $X$ follows the Gaussian distribution and $U$ the Laplace distribution, {\it i.e.},

\begin{align*}
p_x(x; \mu_x, \sigma_x) &= \frac{1}{\sqrt{2\pi\sigma_{x}^{2}}} \exp\left(-\frac{(x-\mu_x)^2}{2\sigma_{x}^{2}}\right)\\ 
p_u(w; \sigma) &= \frac{\sqrt{2}}{2 \sigma}
                 \exp\left(-\frac{\sqrt{2}|u|}{\sigma}\right).
\end{align*}
Note that $\Var(U) = \sigma^2$.

We have three parameters to estimate, $\mu_x$, $\sigma_x$, and $\sigma$.
Since $\mu_x$ can be estimated by the sample average of $W$, we estimate the two standard deviations by the algorithm presented.
The Q-function is written as follows.  
\[
\begin{split}
Q(\theta, \theta'| w_i, d_i) &= 
 d_i \int_{c}^{\infty} h_i(\theta') 
  \left[ \log p_x(x; \sigma_x) + \log p_u(w_i-x; \sigma) \right] dx\\ +
& (1-d_i) \int_{-\infty}^{c} h_i(\theta') 
  \left[ \log p_x(x; \sigma_x) + \log p_u(w_i-x; \sigma) \right] dx.
\end{split}
\] 
Note that the $h$ function is obtained by substituting the density functions to \eqref{eq:h1} and \eqref{eq:h2}.
Setting $\sum_{i=1}^{n}\frac{\partial Q(\theta, \theta'| w_i, d_i) }{ \partial \theta }= 0$ yields the first order conditions for the parameters:
\begin{align}
\sigma &= \frac{\sqrt{2}}{n}
\sum_{i=1}^{n} \Bigg\{
d_i \int_{c}^{\infty} h_i(\theta') |w_i-x| dx +
(1-d_i) \int_{-\infty}^{c} h_i(\theta') |w_i-x| dx 
\Bigg\} 
\label{eq:focsigma} \\
\sigma_{x}^{2} &= \frac{1}{n}
\sum_{i=1}^{n} \Bigg\{
d_i \int_{c}^{\infty} h_i(\theta') (x-\mu_x)^2 dx +
(1-d_i) \int_{-\infty}^{c} h_i(\theta') (x-\mu_x)^2 dx 
\Bigg\}.
\label{eq:focsigmax}
\end{align}

The expressions inside the integral are the weighted average of $|w_i-x|$ and $(x-\mu_x)^2$ respectively (with $h_i(\theta')$ as weights), as analogous to the variance estimator for the Laplace and the Gaussian distributions.
Thanks to the explicit formulas \eqref{eq:focsigma} and \eqref{eq:focsigmax}, the parameters can be updated at each iteration without relying on a numerical optimization routine.
This reduces the computation time and enhances the stability of the algorithm.
Analogous formulas can be obtained for many other parametric distributions, particularly for those belonging to the exponential family.
\end{example}

%
%

\section{Simulation}

This section reports the result of the simulation experiments of the estimators introduced in the previous section.
The methods have been implemented as an R library and are freely available on the GitHub repository (\url{https://github.com/kota7/rddsigma}).

We generate data from various combinations of distributions to examine the robustness of the estimators against misspecification.
$X$ has been generated from the Gaussian distribution and the exponential distribution, while $U$ has been generated from the Gaussian and Laplace distribution.  
For each pair of distributions, we set the variance of $X$ to one, and the variance of $U$, $\sigma$, to 0.2 and 1.2.
The sample size is 500, and the cutoff point $c$ is set to one for all cases.
As a result, we have eight simulation configurations, as summarized in Table \ref{tab:setup}.

\begin{table}[ht] 
\begin{center}
\caption{Simulation setup}
\label{tab:setup}
\begin{tabular}{c|ccccccc}
\hline
ID & N & $c$ & $p_x$ & $p_u$ & $\E(X)$ & $\Var(X)$ & $\Var(U)$ \\
\hline\hline 
1 & 500 & 1 & Gaussian & Gaussian & 0 & 1 & 0.2 \\
2 & 500 & 1 & Gaussian & Laplace & 0 & 1 & 0.2 \\
3 & 500 & 1 & Exponential & Gaussian & 1 & 1 & 0.2 \\
4 & 500 & 1 & Exponential & Laplace & 1 & 1 & 0.2 \\

5 & 500 & 1 & Gaussian & Gaussian & 0 & 1 & 1.2 \\
6 & 500 & 1 & Gaussian & Laplace & 0 & 1 & 1.2 \\
7 & 500 & 1 & Exponential & Gaussian & 1 & 1 & 1.2 \\
8 & 500 & 1 & Exponential & Laplace & 1 & 1 & 1.2 \\
\hline
\end{tabular}
\end{center}
\end{table}

For each setup, we generate 200 random datasets.
Using a generated dataset, we estimate $\sigma$ and other parameters by three methods:
(A) Gaussian-Gaussian estimator, (B) non-Gaussian estimator with $X$ and $U$ following the Gaussian distribution, and (C) non-Gaussian estimator with $X$ following the Gaussian and $U$ following the Laplace distribution.
Notice that for many cases the models are ``misspecified'' in a sense that the true data generating process does not follow the distributions assumed by the model.  This allows us to examine the robustness of the estimators against the deviation from the assumptions.

The results are summarized in Figure \ref{fig:sim}. The numbers in the horizontal axis correspond to the IDs given in Table \ref{tab:setup} and each panel corresponds to an estimation method.
The Gaussian-Gaussian estimator, labeled as (A), consistently recovers the true parameter for all cases.  IDs 1 and 5 satisfy the model assumptions and estimated $\sigma$ distributes around the true parameters as expected.  Even for other cases where the model is misspecified, the estimates are centered around the true parameter.

The estimator (B), the non-Gaussian estimator with the assumption that $X$ and $U$ follow the Gaussian distribution, also estimates the parameters correctly in most cases.  It tends to be, however, unstable for the setups 3 and 7, where the distribution of $X$ is generated from the exponential distribution.  

The estimator (C), the non-Gaussian estimator with the assumption that $X$ follows the Gaussian and $U$ follows the Laplace distribution, performs well for IDs 2 and 6, which satisfy the model assumptions.
However, it exhibits a relatively high sensitivity to misspecification compared with the other two methods.  
Instability is particularly prominent for the cases where $X$ is generated from the exponential distribution.  
 
\begin{figure}[ht] 
\begin{center}
\includegraphics[width=0.95\hsize]{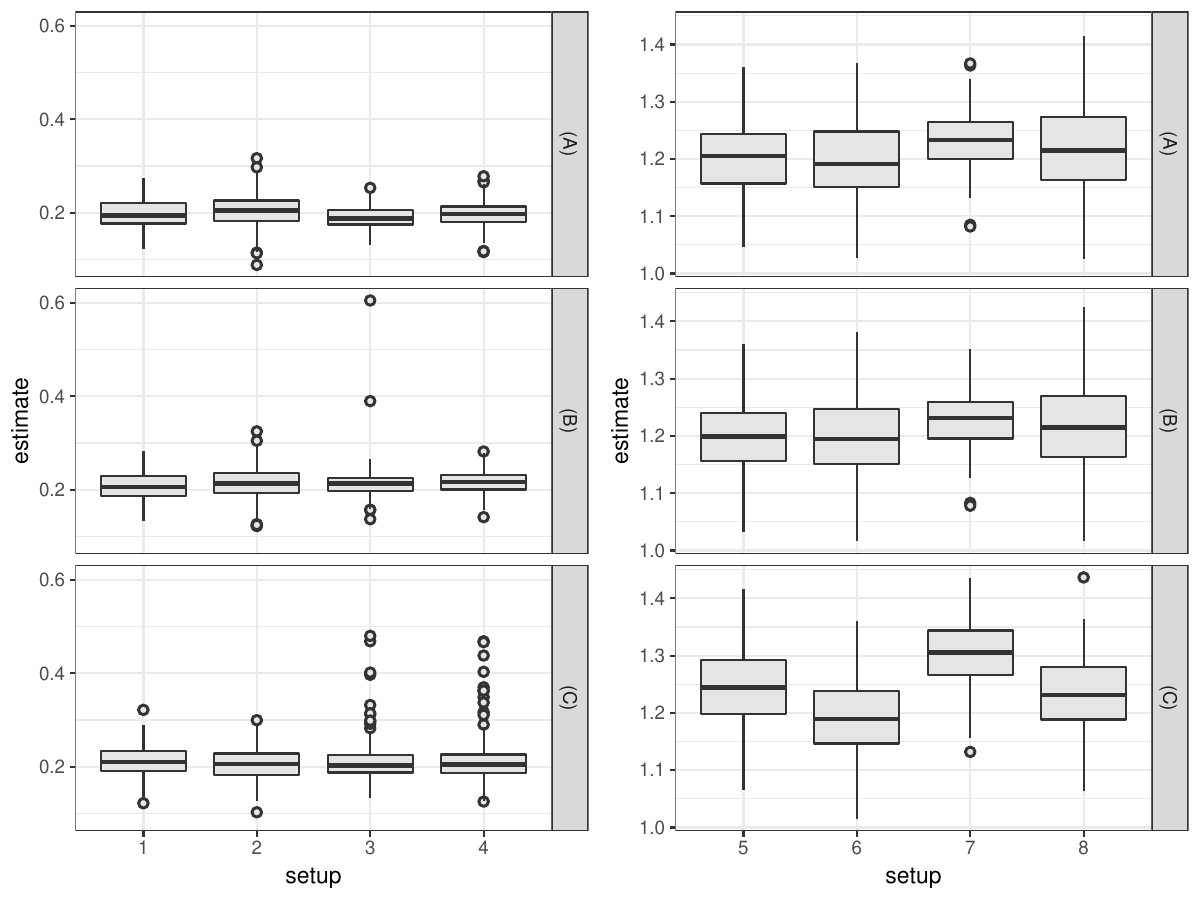}
\caption{Distribution of estimated $\sigma$.  Each boxplot comprises 200 independent estimates.  The numbers in the horizontal axis indicate the IDs of the data generating process given in Table \ref{tab:setup}. True parameters are 0.2 in the panels in the left column and 1.2 in the right. Each row uses a different estimation method: (A) Gaussian-Gaussian estimator, (B) non-Gaussian estimator with $X$ and $U$ following the Gaussian distribution, and (C) non-Gaussian estimator with $X$ following the Gaussian and $U$ following the Laplace distribution.}
\label{fig:sim}
\end{center}
\end{figure}

\section{Concluding Remarks}

This paper introduces two estimators for estimating the variance of measurement errors in running variables of sharp regression discontinuity designs.
The first estimator is constructed under the assumption that both the running variable and the measurement error follow the Gaussian distribution.
Under this assumption, the conditional likelihood function has an explicit formula, and the parameters can be estimated efficiently using a numerical optimization routine.
Despite the strong assumptions on the variable distributions, the estimator exhibits robustness against misspecification in the simulation exercises.

The second estimator relaxes the Gaussian assumption and allows both $X$ and $U$ to follow arbitrary distributions characterized by a finite number of parameters.
A variant of the expectation-maximization (EM) algorithm is introduced, which optimizes the likelihood function efficiently compared with a direct application of a standard numerical optimization routine.
This estimator performs as well as the first when the model is correctly specified.
However, the simulation experiments find that the estimator can become unstable and biased when the model assumptions deviates from the data generating process.

The first estimator would be practical in many cases for estimating the variance of measurement errors. 
It is easy to implement, is computationally efficient, and tends to be robust against misspecification. 
The second estimator can be preferred in domains where the distributions of the variables are understood well.
It would also serve as a robustness check for the first estimator.



\appendix
\appendix

\section{Proof} \label{sec:proof}

We provide a proof for the inequality \eqref{eq:ineq}:
\[
\log p(W, D; \theta) - \log p(W, D; \theta') \ge
Q(\theta, \theta' | W, D) - Q(\theta', \theta' | W, D).
\]
To do so, we introduce the following lemma.

\begin{lemma*}
Let $J(\theta) = \log \int_{x \in \mathcal{X}} g(x; \theta) dx$, where $g$ is a positive-valued function and $\mathcal{X}$ is a subset of the range of $g$.  Define the corresponding Q-function by
\[ Q(\theta, \theta') = \int_{x \in \mathcal{X}} h(x; \theta') \log g(x; \theta) \]
where 
\[ h(x; \theta) = \frac{g(x; \theta)}{\int_{y \in \mathcal{X}} g(y; \theta) dy}. \]
Then, 
\[
\log J(\theta) - \log J(\theta') \ge Q(\theta,\theta') - Q(\theta', \theta').
\]
\end{lemma*}

\begin{proof}
\begin{align*}
& \log J(\theta) - Q(\theta, \theta')\\ 
=& \log \int_{x \in \mathcal{X}} g(x; \theta)dx - \int_{x \in \mathcal{X}} h(x; \theta') \log g(x; \theta)dx \\
=& \int_{y \in \mathcal{X}} h(y; \theta') \log \int_{x \in \mathcal{X}} g(x; \theta)dx dy - \int_{x \in \mathcal{X}} h(x; \theta') \log g(x; \theta)dx \\
=& \int_{y \in \mathcal{X}} h(y; \theta') \left( \log \int_{x \in \mathcal{X}} g(x; \theta)dx - \log g(y; \theta) \right) dy \\
=&  \int_{y \in \mathcal{X}} h(y; \theta') \log \frac{\int_{x \in \mathcal{X}} g(x; \theta)dx}{g(y; \theta)} dy \\
=& - \int_{y \in \mathcal{X}} h(y; \theta') \log h(y; \theta) dy \\
=& - \int_{x \in \mathcal{X}} h(x; \theta') \log h(x; \theta) dy.
\end{align*}
Construct the same equality with $\theta=\theta'$ and subtract from the both sides, then
\begin{align*}
& \log J(\theta) - \log J(\theta') - Q(\theta, \theta') + Q(\theta', \theta')  \\
=& \int_{x \in \mathcal{X}} h(x; \theta') \log \frac{h(x; \theta')}{h(x; \theta)} \\
\ge& 0
\end{align*}
where the last line is due to the Gibb's inequality.
Hence,
\[
\log J(\theta) - \log J(\theta') \ge Q(\theta,\theta') - Q(\theta', \theta').
\]
\end{proof}

To derive the inequality \eqref{eq:ineq}, apply the lemma with $g(x; \theta) = p_x(x; \theta_x)p_u(W-x; \theta_u)$ and $\mathcal{X} = (c, \infty)$.  Then, we obtain 
\begin{equation} \label{eq:d1}
 \log \int_{c}^{\infty} p_x(x; \theta_x)p_u(W-x; \theta_u) \ge \int_{c}^{\infty} h(\theta' | Z, D) 
  \left( \log p_x(x; \theta_x) + \log p_u(W-x; \theta_u) \right) dx.
\end{equation}
Similarly, applying the lemma with the same $g$ function and $\mathcal{X} = (-\infty, c)$, 
\begin{equation} \label{eq:d0}
\log \int_{-\infty}^{c} p_x(x; \theta_x)p_u(W-x; \theta_u) \ge \int_{-\infty}^{c} h(\theta' | Z, D) 
  \left( \log p_x(x; \theta_x) + \log p_u(W-x; \theta_u) \right) dx.
\end{equation}
\eqref{eq:d1} and \eqref{eq:d0} imply \eqref{eq:ineq}.

\bibliography{rdd-me.bib}
\end{document}